\documentclass[a4paper]{article}

\usepackage[margin=1in]{geometry}
\usepackage[english]{babel}
\usepackage[utf8]{inputenc}
\usepackage{amsmath}
\usepackage{amsthm}
\usepackage{amssymb}
\usepackage{graphicx}
\usepackage[colorinlistoftodos]{todonotes}
\usepackage{enumerate}
\usepackage{scrextend}
\usepackage[numbers]{natbib}
\usepackage[]{algorithm2e}
\usepackage{array}
\usepackage{url}

\newcommand\Mark[1]{\textsuperscript#1}

\usepackage{tikz}

\usetikzlibrary{decorations.markings}
\tikzstyle{vertex}=[circle, draw, inner sep=0pt, minimum size=6pt]

\newtheorem{theorem}{Theorem}[section]
\newtheorem{lemma}[theorem]{Lemma}
\newtheorem{corollary}[theorem]{Corollary}
\newtheorem{prop}[theorem]{Proposition}
\theoremstyle{definition}
\newtheorem*{define}{Definition}

\theoremstyle{remark}
\newtheorem{remark}[equation]{Remark}

\date{\today}

\begin{document}

\renewcommand{\abstractname}{\large Abstract}

\begingroup
\centering
{\LARGE Discrete All-Pay Bidding Games }\\[1.5em]
\large Michael Menz\Mark{1}, Justin Wang\Mark{2}, Jiyang Xie\Mark{3}\\[1em]
\begin{tabular}{*{3}{>{\centering}p{.3\textwidth}}}
\Mark{1}Yale University & \Mark{2}Yale University & \Mark{3}Yale University \tabularnewline
\url{michael.menz@yale.edu} & \url{justin.wang@yale.edu} & \url{jiyang.xie@yale.edu}
\end{tabular}\par
\endgroup

\vspace{1cm}

\begin{abstract}
In an all-pay auction, only one bidder wins but all bidders must pay the auctioneer. All-pay bidding games arise from attaching a similar bidding structure to traditional combinatorial games to determine which player moves next. In contrast to the established theory of single-pay bidding games, optimal play involves choosing bids from some probability distribution that will guarantee a minimum probability of winning. In this manner, all-pay bidding games wed the underlying concepts of economic and combinatorial games. We present several results on the structures of optimal strategies in these games. We then give a fast algorithm for computing such strategies for a large class of all-pay bidding games. The methods presented provide a framework for further development of the theory of all-pay bidding games.
\end{abstract}

\section{Introduction}

At the conclusion of an all-pay auction, all bidders must pay the bids they submitted, with only the highest bidder receiving the item.  With this idea in mind, one can play a variant of a two-player game using an all-pay auction to decide who moves next instead of simply alternating between players.  For example, one could play all-pay Tic-Tac-Toe with $100$ chips. Each round both players privately record their bids and then simultaneously reveal them. If player A bids $40$ and his opponent bids $25$, player A would get to choose a square to mark and the next round of bidding would begin wih player $A$ having $85$ chips, player $B$ having $115$ chips.  Note that the chips have no value outside the game and only serve to determine who moves - the ultimate goal is still just to get three-in-a-row.

Another variant of the game could have only the player who wins the move pay his/her bid, i.e. deciding who moves next using a first-price auction.  These games were first studied formally in the 1980s by Richman, whose work has since then been greatly expanded upon.  
Intuitively, there is less risk in these ``Richman games'' for the player losing the bid.  If your opponent bids $100$ for a certain move, it makes no difference whether your bid was $99$ or $0$.  All that matters is that your opponent's bid was higher.  A surprising consequence of this single-pay structure is that for every state of a game, there exists a ``Richman value'' $v$ for each player that represents the proportion of the total chips that player would need to hold to have a deterministic winning strategy.  In this situation, the player with the winning strategy can tell her opponent what bid she will be making next without affecting her ability to ultimately win.  For zero-sum games, 
this means that unless a player's chip ratio is exactly $v$, then one of the players must have such a winning strategy \citep{LLPU96}, \citep{LL+99}. 

Our objective is to begin the formal study of all-pay bidding games. Returning to the above example where your opponent bids $100$ chips and you are indifferent between bidding $99$ and $0$, it is clear this is no longer true for an all-pay bidding mechanism. You would be very disappointed had you bid $99$, as your opponent would be paying just $1$ chip on net to make a move.  Had you bid $0$, though, you might feel pretty good about not moving this current turn, as the $100$ extra chips may make a bigger difference for the rest of the game.  Thus, there are at least two bidding scenarios which intuitively seem like very good positions to be in: winning the bid by a relatively small number of chips or losing the bid by a relatively large number of chips.  This behavior suggests that, unlike in Richman games, in all-pay bidding games one of the players will not necessarily have a deterministic winning strategy.  Instead, players must randomize their bidding in some way. Thus, we must appeal to the concept of mixed bidding strategies in Nash equilibria.

\subsection{A Game of All-Pay Bidding Tic Tac Toe}
Before presenting formal definitions and results, we provide a sample all-pay bidding game to illustrate some of the main features of playing these games.
Alice and Bob, each with $100$ chips, are playing all-pay bidding Tic-Tac-Toe.
Each turn Alice and Bob secretly write down a bid, a whole number less than or equal to their total number of chips. 
They then reveal their bids and whoever bid more gets to decide who makes the next move. 
We say a player has \textbf{advantage} if, when players bid the same amount, that player gets to decide who makes the next move. 
The question of deciding how to assign advantage is one we encountered early on. For our games, we give advantage to the player with more chips, then arbitrarily let Alice have advantage when Alice and Bob have the same number of chips. A number of other mechanisms would also suffice, such as alternating advantage or having a special ``tie-breaking'' chip that grants advantage and is passed each time it is used. Our choice was made in the interest of computational simplicity and to eventually allow extension to real-valued bidding.  

\textit{First Move.} Both players have $100$ chips. Alice bids $25$, Bob bids $40$. Bob wins the right to move and plays in the center of the board.

\begin{center}
\includegraphics[width=0.7in]{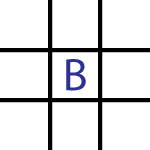}
\end{center}

\textit{Second Move.} Alice has $115$ chips, Bob has $85$ chips. Alice wants to win this move to keep pace with Bob, but also does not see why it should be worth more than the first, so she only slightly increases her bid to $30$. Bob, thinking that Alice may want to win this move more, is content to let Alice win and collect chips by bidding $0$. Alice wins the right to move and plays in the top-left corner of the board.

\begin{center}
\includegraphics[width=0.7in]{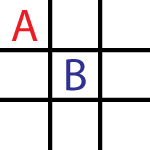}
\end{center}

\textit{Third Move.} Alice has $85$ chips, Bob has $115$ chips.  Alice bids $45$, Bob bids $40$, so Alice wins the right to move and plays in the top-center of the board.

\begin{center}
\includegraphics[width=0.7in]{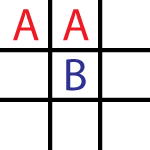}
\end{center}

\textit{Fourth Move.} Alice has $80$ chips, Bob has $120$ chips.  Alice is one move away from winning and decides to risk it and bid all of her $80$ chips. Unfortunately for her, Bob has guessed her move and has himself bid $80$ as well. Because Bob has more chips overall, he uses his advantage to win the tie and plays in the top-right corner of the board, blocking Alice's victory and setting himself up for one.

\begin{center}
\includegraphics[width=0.7in]{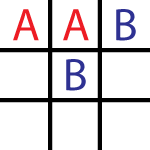}
\end{center}

\textit{Fifth Move.} Alice has $80$ chips, Bob has $120$ chips. Bob has more chips and is just a move away from winning, so he can bid everything, play in the bottom-left corner and win the game.

In normal Tic-Tac-Toe, both players can guarantee a draw by playing well, but as we see from this example, the result of a game of all-pay Tic-Tac-Toe involves far more chance.

For example, at the fourth move in the above game, Alice could have guessed Bob might bid $80$ and chosen to ``duck'' by bidding $0$. In this case Bob would win the move and play as before, but now the chip counts would be $160$ to $40$ in Alice's favor, and Alice can bid $40$ and then $80$ to win the next two moves and win in the left column. It is easy to see that if a player knows what his opponent will bid at each move, he can win the game easily. Thus, in the vast majority of all-pay bidding games, optimal play cannot be deterministic.

Though we do not return to Tic-Tac-Toe in this paper, it served as a test game for much of our research. Using our results, we built a computer program to play all-pay bidding Tic-Tac-Toe optimally. The program can be played against at \url{http://biddingttt.herokuapp.com}. The theory behind this program, which is not specific to Tic-Tac-Toe, will be the focus of the rest of the paper.

\subsection{Overview of Results}

Our ultimate goal is to characterize the optimal strategies for a general class of all-pay bidding games.
The game consists of iterations of both players bidding for the right to move followed by one of the players making a move. In turn, an optimal strategy will also have two parts: the bid strategy and the move strategy.  For a given position in the game (e.g. a configuration of the Tic-Tac-Toe board) and chip counts for each of the players (e.g. Alice has $115$ chips, Bob has $85$ chips), the bid strategy must tell players how to best randomize their bets (e.g. Alice bids $0$ chips half the time, $80$ chips half the time) while the move strategy must tell whoever wins the bid the best move to make (e.g. where to play on the Tic-Tac-Toe board).

The problem of determining move strategy is largely combinatorial in nature and remains similar to its analog in Richman games. We can still represent the space of game states as a directed graph, and there is a not always a single best move that each player can make upon winning the bid.  That is, the best move could also depend on each player's chip counts moving forward.

The focus of our work, then, will be on determining the optimal bidding strategy for any game position and chip counts.
Naturally, this should depend on a player's chances of winning in any of the possible subsequent game situations (i.e. after a single move and updated chip counts).  For purposes of initial analysis, we will assume that these future winning probabilities are known, and see how the bidding strategy can be determined from this information.  Then, by using the recursive nature of the directed graph, we will be able to start from the ``win'' and ``loss'' nodes (where the probability is just $1$ or $0$) to find the optimal bidding strategies and winning probabilities for any game situation.  For the rest of this paper, we will often refer to a bidding strategy as just a ``strategy'' when it is clear that the focus is just on the bidding side of the game.  Here, a strategy will be a probability vector where the $i$th coordinate corresponds to the probability a player will bid $i$ chips. Further, a Nash equilibrium for a game situation will just be a pair of strategies so that neither player has an incentive to deviate.  This means that each player's strategy will maximize his/her minimum probability of ultimately winning from the next turn of the game.

It quickly becomes apparent that a naive recursive algorithm using linear programming is feasible only for games with very few moves.  Thus, in the interest of being able to practically calculate the optimal bidding strategies for general games, we prove some structural results on the Nash equilibria.  In particular, useful structure arises when we study a particular class of games that we dubbed ``precise'', which roughly speaking are games where having one more chip is strictly better than not.  The key result is a surprising relationship between opposing optimal strategies that allows one to immediately write a Nash equilibrium strategy for the player without advantage if given a Nash equilibrium strategy for the player with advantage.

This relationship, (\ref{reverse_thm}), which we call the Reverse Theorem, is a critical step toward the calculation of optimal strategies for precise games. Further, by assigning an arbitrarily small value in the game to each chip, we get a precise game that is very similar to the original game.  We show that the optimal strategies we can calculate for these new precise games will indeed converge to optimal strategies for our possibly imprecise games. Our theoretical results ultimately culminate in a fast algorithm for computing optimal probabilistic bidding strategies. Together with a move strategy for the combinatorial side of the game, this gives a complete characterization of optimal play for all-pay bidding games.

\section{Strategies in precise games}\label{tools}

Let $G_{a, b}$ denote a single turn of a two-player all-pay bidding game $\mathcal{G}$ where player $A$ is endowed with $a$ chips and player $B$ is endowed with $b$ chips. The underlying combinatorial game $\mathcal{G}$ is a two-player zero-sum game, represented by an acyclic, colored, directed graph with two marked vertices, $\mathcal{A}$ and $\mathcal{B}$. The game begins by placing a token at some starting vertex. At each turn, a player moves the token to an adjacent vertex. Player $A$ wins if the token reaches $\mathcal{A}$ and player $B$ wins if the token reaches $\mathcal{B}$. By saying the graph is colored, this means that edges are one of two colors, say red and blue, such that $A$ can only move the token along red edges and $B$ can only move the token along blue edges.  To ensure consistency in the bidding strategy from turn to turn, we seek to avoid situations where the winner of a bid can be put in zugzwang - i.e. where it would be better to not move at all.  Thus, the bid winning player, rather than simply being able to move next, gets to determine who moves next.  With this condition, $\mathcal{G}$ can an asymmetric game where zugzwang is possible, like chess and many other popular two player games.

The \textbf{payoff}, or value of the game, for player A at $G_{a, b}$ is denoted by $v_A(G_{a, b}) \in [0,1]$ and is equal to the probability that player A wins the game under optimal play. That is, we set $v_A(\mathcal{A}) = 1$ and $v_A(\mathcal{B}) = 0$ and calculate payoffs recursively. Similarly, let $v_B(G_{a, b})$ denote the probability that player $B$ wins the game. Often, when the chip counts or specific combinatorial game are not relevant to the discussion, the payoffs will be shortened to $v_A$ and $v_B$. Note that $v_B = 1-v_A$ as we only study games that cannot end in ties (for the game of Tic-Tac-Toe above, we can arbitrarily let one of the players win all draws).

Thus a \textbf{payoff matrix} for player A in $G_{a, b}$ is denoted by $M_A(G_{a, b})$ and is given by
\[ (M_A)_{i,j} = 
\left\{ \begin{array}{cc} 
\max(\max_{G' \in \mathcal{S}_A(G)} v_A(G'_{a-j+i, b-i+j}), \min_{G' \in \mathcal{S}_B(G)} v_A(G'_{a-j+i, b-i+j})) & \text{if $A$ wins the bid}\\
\min(\min_{G' \in \mathcal{S}_B(G)} v_A(G'_{a-j+i, b-i+j}), \max_{G' \in \mathcal{S}_A(G)} v_A(G'_{a-j+i, b-i+j})) & \text{if $B$ wins the bid}\end{array}\right.
\]
where $\mathcal{S}_A(G)$ and $\mathcal{S}_B$ are the set of game positions that can be moved to from $G$ by $A$ and $B$ respectively. 
The $(i, j)$ entry corresponds to player $A$'s probability of winning the game after $A$ bids $j$ and $B$ bids $i$.
Note this is well-defined because the game is zero-sum: by moving to the game state that minimizes Player $A$'s payoff, player $B$ is maximizing his own payoff at the same time (and vice-versa).
Similarly, let $M_B(G_{a, b})$ denote the payoff matrix for player $B$.

We notice that if player $A$ bids $x$ and player $B$ bids $y$, this is equivalent to player $A$ bidding $x+z$ and player $B$ bidding $y+z$ for any $z$ because the players are paying each other. Thus, we have that payoff matrices are Toeplitz, or diagonal-constant. We will write player $A$'s and player $B$'s payoff matrices for $G_{a, b}$ as
\[
\left( \begin{array}{cccc}
\alpha_0 & \alpha_1 & \ldots & \alpha_a \\
\alpha_{-1} & \alpha_0 & \ldots & \alpha_{a-1} \\
 \vdots & \vdots & \vdots & \vdots \\
\alpha_{-b} & \alpha_{-b+1} & \ldots & \alpha_{-b+a}
\end{array} \right)
\hspace{1cm}\text{and}\hspace{1cm}
\left( \begin{array}{cccc}
\beta_0 & \beta_1 & \ldots & \beta_b \\
\beta_{-1} & \beta_0 & \ldots & \beta_{b-1} \\
\vdots & \vdots & \vdots & \vdots \\
\beta_{-a} & \beta_{-a+1} & \ldots & \beta_{-a+b}
\end{array} \right)\]
respectively.

We pause to consider a simple example. Let the underlying game be one where player $A$ needs to make two moves to win, while player $B$ needs to make only one more move to win. Suppose player $A$ has $5$ chips while player $B$ has $3$ chips. Then we would get the following payoff matrices for player $A$ and player $B$ 
\[
\left( \begin{array}{cccccc}
1 & 1 & 1 & 0 & 0 & 0\\
0 & 1 & 1 & 1 & 0 & 0\\
0 & 0 & 1 & 1 & 1 & 0\\
0 & 0 & 0 & 1 & 1 & 1
\end{array} \right)
\hspace{1cm}\text{and}\hspace{1cm}
\left( \begin{array}{cccc}
0 & 1 & 1 & 1 \\
0 & 0 & 1 & 1 \\
0 & 0 & 0 & 1 \\
1 & 0 & 0 & 0 \\
1 & 1 & 0 & 0 \\
1 & 1 & 1 & 0
\end{array} \right)\]
respectively

A \textbf{strategy} for player A in $G_{a, b}$ is given by an $(a+1)$-dimensional column vector with all non-negative entries that sum to 1. The $i$-th entry of this vector (where we start indexing at $0$) gives the probability that player $A$ will bid $i$ chips. Similarly, a \textbf{strategy} for player $B$ in $G_{a,b}$ and is given by a $(b+1)$-dimensional column vector satisfying the same conditions. We denote a Nash equilibrium strategy in the game $G_{a,b}$ as $S_A(G_{a,b})$ for player $A$ and as $S_B(G_{a,b})$ for player $B$. Often times we will not be too explicit with the size of these vectors. It should be clear from context.

Note that the $i$th row of $M_A$ corresponds to the payoffs of each of $A$'s pure strategies if her opponent $B$ bids $i$.
Letting $A_i$ be the $i$th row of $M_A$, we have $A_i \cdot S_A= a_{i0}(S_A)_0 + \cdots + a_{ia}(S_A)_a= (P_A)_i$, a weighted average of $A$'s pure payoffs when $B$ bids $i$.
Thus, $(P_A)_i$ is player $A$'s probability of winning if her strategy is $S_A$ and her opponent purely bids $i$.
For example, if we have
\[
M_A\cdot S_A= \left( \begin{array}{cc}
1 & \frac{1}{2} \\
0 & 1
\end{array}\right) \cdot \left(\begin{array}{c}
\frac{1}{2} \\
\frac{1}{2}
\end{array}
\right)=
\left(\begin{array}{c}
\frac{3}{4} \\
\frac{1}{2}
\end{array}
\right),
\]
this means by playing $S_A$, player $A$ wins $\frac{3}{4}$ of the time if player $B$ only bids $0$ and wins $\frac{1}{2}$ of the time if player $B$ only bids $1$.

Now, if player $B$'s strategy is $S_B$, $S_B^TM_AS_A=(S_B)_0(P_A)_0+\cdots+(S_B)_b(P_A)_b$, another weighted average of $A$'s payoffs for each of $B$'s pure strategies.
Thus, $S_B^TM_AS_A$ is exactly $A$'s payoff if she plays $S_A$ and her opponent plays $S_B$.
$S_A^TM_BS_B$ is $B$'s payoff in the same situation.
Continuing with the above example, if we now let $S_B^T=\left(\begin{array}{cc}
\frac{1}{2}&
\frac{1}{2}
\end{array}\right)$
then $S_B^TM_AS_A=\frac{1}{2}\cdot\frac{3}{4}+\frac{1}{2}\cdot\frac{1}{2}=\frac{5}{8}$.
So given strategies $S_A$ and $S_B$ for players $A$ and $B$,
player $A$ has a $\frac{5}{8}$ probability of winning.

We compile these results in the lemma below.

\begin{lemma}[]\label{basics}
Let $ M_A $ and $ M_B $ be payoff matrices for players $ A $ and $ B $, respectively, in $ G_{a, b} $. Then the following statements are true.
\begin{itemize}
\item[(a)] The diagonals of $ M_A $ and $ M_B $ are constant, i.e. the payoff matrices are Toeplitz.

\item[(b)] Let $\mathbf{1}$ be the appropriately sized matrix whose entries are all 1. Then $ M_B =\mathbf{1} -M_A ^ T $.

\item[(c)] Suppose $ (S_A, S_B) $ is a Nash equilibrium. Then $ (M_B S_B)_i = v_B $ if $ (S_A)_i\neq 0 $ and $ (M_A S_A)_i = v_A $ if $ (S_B)_i\neq 0 $.
\end{itemize}
\end{lemma}

\noindent This lemma provides the basic structure from which many of our main proofs will follow from later.

It is clear that $v_A(G_{a+1, b}) \ge v_A(G_{a, b})$, since player $A$ can always bid as if he did not have the extra chip. We now define a class of games pivotal to our analysis in which this inequality is strict. Formally, a game $G$ is called \textbf{precise} if in every successor state to $G$, it is strictly better to have one more chip.

\begin{remark}\label{precise_matrices}
We note that in particular, this guarantees a certain strict monotonicity among the entries of the payoff matrices. In particular, winning the bid by one less chip is always strictly preferable, as is losing by one more chip. Thus we have that for the player with advantage, $\alpha_i > \alpha_j$ for $0 \le i < j$ and $\alpha_i > \alpha_j$ for $i < j < 0$. A similar relationship holds for the player without advantage, except $\beta_0 < \beta_1$ and $\beta_0 > \beta_{-1}$. 
\end{remark}

\begin{define}
 A strategy $S$ has \textbf{length} $\ell=\ell(S)$ if $S_{\ell-1} \ne 0$ and $S_{m} = 0$ $\forall m \geq \ell$. A strategy $S$ is \textbf{gap-free} if $S_i,S_j \ne 0$ if and only if $S_k \ne 0$ $\forall i \le k \le j$.
\end{define}

The definition of length encapsulates the observation that unless the game is close to completion, players will never bid a large proportion of their chips. The second definition seems more arbitrary at the moment, but it plays a pivotal role in the following Proposition and will serve to greatly simplify the language throughout the paper.

\begin{prop}[]\label{grounded_gapfree}
Let $ G_{a, b} $ be precise. Any equilibrium strategy for the player with advantage is gap-free and bids $0$ with nonzero probability, while any equilibrium strategy for the other player is gap-free and bids 1 with nonzero probability. If the player with advantage has an equilibrium strategy of length $\ell  $, any equilibrium strategy for the other player has length $\ell  $ or $\ell  +1 $.
\end{prop}

\begin{proof}
Suppose without loss of generality that player $ A $ has advantage, and let $ S_A = (s_0,\dots, s_a) $ and $ S_B = (t_0,\dots , t_b) $ be equilibrium strategies for players $ A $ and $ B $ respectively. We claim that if $ i\geq 0 $,
\begin{itemize}
\item[(i)] $ s_i = 0 $ implies $ t_{i +1} = 0 $, and
\item[(ii)] $ t_{i +1} = 0 $ implies $ s_{i +1} = 0 $.
\end{itemize}
If $ s_i = 0 $ and $ t_{i +1}> 0 $, player $ B $ should alter his strategy so that he bids $ i $ with probability $ t_i+ t_{i +1} $ and $ i +1 $ with probability 0. This saves player $ B $ a chip whenever he would have bid $ i +1 $ without changing any possible outcome of these bids, and all other possibilities are unchanged. By precision, this new strategy is strictly better than $ S_B $ for player $ B $, a contradiction. This proves (i).

If $ t_i = 0 $ and $ s_i> 0 $, player $ A $ should alter her strategy  so that she bids $ i $ with probability $ s_i+ s_{i +1} $ and $ i +1 $ with probability 0. As in the previous case this new strategy is strictly better for player $ A $, a contradiction, proving (ii).

Together, (i) and (ii) complete the proof except in the case when $ S_B = (1, 0,\dots , 0) $. However, in this case an optimal strategy for player $ A $ is to also bid 0 with probability 1, and it follows that $G_{a, b} $ is not precise.
\end{proof}

This characterization of equilibrium strategies is what motivated our restriction to precise games. In the presence of precision, an easily observable, yet highly unexpected relationship between opposing optimal strategies appears. This relationship forms the foundation for the rest of our results.

\begin{define}
The \textbf{reverse} of a length $\ell$ strategy $S$ is given by 
\[ \mathcal{R}(S) = \mathcal{R}((s_0,s_1,\ldots,s_{\ell-1},0,\ldots,0)) = (s_{\ell-1},s_{\ell-2},\ldots,s_0,0,\ldots,0). \]
where the number of trailing zeroes will be clear from context.
\end{define}

\begin{theorem}[]\label{reverse_thm}
Suppose that $ G_{a, b} $ is precise, and that $ S $ is an equilibrium strategy for the player with advantage. Then $\mathcal{ R} (S) $ is an equilibrium strategy for the player without advantage.
\end{theorem}
\begin{proof}
Suppose without loss of generality that player $ A $ has advantage, and $ S = S_A = (s_0,\dots , s_{\ell -1}, 0, \dots, 0) $ has length $ \ell  $. By Lemma \ref{basics} and Proposition \ref{grounded_gapfree}, we have

\begin{equation}\label{reverse_thm1}
M_A \cdot S_A=
\left[ \begin{array}{cccc}
\alpha_0 & \alpha_1 & \ldots & \alpha_a \\
\alpha_{-1} & \alpha_0 & \ldots & \alpha_{a-1} \\
 \vdots & \vdots &  & \vdots \\
\alpha_{-b} & \alpha_{1-b} & \ldots & \alpha_{ a-b}
\end{array} \right]
\left [ \begin{array}{c}
s_0 \\
s_1\\
\vdots \\
s_{\ell-1} \\
0 \\
\vdots \\
0\\
\end{array} \right]
= 
\left [ \begin{array}{c}
w_0 \\
v_A \\
\vdots \\
v_A \\
w_{\ell} \\
\vdots \\
w_{b}
\end{array} \right],
\end{equation}
where $ w_0, w_{\ell },\dots , w_b\geq  v_A $. 
We claim further that $ w_0 = v_A $.

Suppose for a contradiction that $ w_0 > v_A $. Then if $ S_B $ is an equilibrium strategy for player $ B $, by Lemma \ref{basics} and Proposition \ref{grounded_gapfree} it is of the form $ S_B = (0, t_1, t_2,\dots , t_\ell , 0,\dots , 0) $ where $ t_1> 0 $, but possibly $ t_\ell = 0 $. 

When played against $ S_A $, $ S_B $ gives a payoff of $ v_B $.
Let $ v_B' $ be player $ B $'s payoff against $ S_A $ when he plays the shifted strategy $ S_B' = (t_1, t_2, \dots , t_\ell , 0,\dots , 0) $. Since $ (S_A, S_B) $ is a Nash equilibrium, $ v_B' \leq v_B $. On the other hand, player $ A $ can guarantee a payoff of $ 1-v_B' $ against $ S_B $ by using the strategy $ S_A' = (0, s_0,s_1,\dots , s_{\ell -2},s_{\ell - 1},\dots,0) $ since the probability of any given difference in bids occurring is the same in $(S_A', S_B) $ as in $ (S_A, S_B') $. Therefore $ v_B'\geq v_B $, so $ v_B = v_B' $, whence $S_B^T \cdot M_A \cdot S_A=S_B'^T \cdot M_A\cdot S_A $. Expanding this, we find
\begin{equation*}
0\cdot w_0 + t_1\cdot v_A+\cdots+t_{\ell - 1} \cdot v_A + t_{\ell}\cdot w_{\ell} =
t_1\cdot w_0+t_2\cdot v_A +\cdots+t_{\ell}\cdot v_A
\end{equation*}
Suppose $w_{\ell} > v_A$. Then, we must have $t_{\ell} = 0$, which solves to get $w_0 = v_A$. If $w_{\ell} = v_A$, the equation solves the same way to get $w_0 = v_A$. Thus, either way we have a contradiction of $w_0 > v_A$. Thus, $w_0 = v_A$.
Together with \eqref{reverse_thm1}, this gives
\begin{equation}\label{reverse_thm2}
v_A =\alpha _0 s_0+\cdots +\alpha _{\ell -1} s_{\ell -1} =\cdots =\alpha _{-(\ell - 1)} s_0+\cdots + \alpha_0 s_{\ell -1}. 
\end{equation}
By Lemma \ref{basics} we have $M_B = \mathbf{1} - M_A^T$, so
\[
M_B \cdot \mathcal{R}(S_A)=
\left[ \begin{array}{cccc}
 1-\alpha_0 & 1-\alpha_{-1} & \ldots & 1-\alpha_{-b} \\
 1-\alpha_{1} & 1-\alpha_0 & \ldots & 1-\alpha_{1-b} \\
 \vdots & \vdots &  & \vdots \\
 1-\alpha_{a} & 1-\alpha_{a-1} & \ldots & 1-\alpha_{a-b}
\end{array} \right]
\left[ \begin{array}{c}
s_{\ell-1} \\
\vdots \\
s_0 \\
0 \\
\vdots \\
0
\end{array} \right].
\]
For $ 0\leq i\leq \ell  -1 $ we have $(1-\alpha_i)s_{\ell-1}+\cdots+(1-\alpha_{i-\ell+1})s_0=(s_0+\cdots+s_{\ell-1})-
(\alpha_{i-\ell+1}s_0+\cdots+\alpha_is_{\ell-1})=1-v_A= v_B$ by equation  \eqref{reverse_thm2}. In other words, $\mathcal{ R} (S_A) $ guarantees player $ B $ his highest possible payoff against $ S_A $, so he has no incentive to deviate from $\mathcal{ R}  (S_A) $ if player $ A $ uses $ S_A $.

We now show player $A$ has no incentive to deviate from $S_A$ against $\mathcal{R}(S_A)$. If $\ell\le i \leq a$, the payoff for player $B$ if player $ A $ bids $i$ will be $(s_0+\cdots+s_{\ell-1})-(\alpha_{i-\ell+1}s_0+\cdots+\alpha_is_{\ell-1})$.
By the formulation of precision in terms of payoff matrices in Remark  \ref{precise_matrices}, we have strict inequalities $\alpha_{i-\ell+1}<\alpha_0$ through $\alpha_i<\alpha_{\ell-1}$, so $\alpha_{i-\ell+1}s_0+\cdots+\alpha_is_{\ell-1} < \alpha_0s_0+\cdots+\alpha_{\ell-1}s_{\ell-1} = v_A$. Thus player $A$ loses utility if she bids any amount greater than $\ell  -1 $ with positive probability. One also readily sees that if she alters her distribution of bids $ 0,\dots ,\ell -1 $ this will not change her payoff against $\mathcal{R}(S_A)$. It follows that $ (S_A,\mathcal{ R} (S_A)) $ is a Nash equilibrium as claimed.
\end{proof}

The Reverse Theorem reveals a strong relationship between opposing player's strategies. Using it, we can now fully characterize the set of optimal strategies for both players in precise games.

\begin{theorem}[]\label{advantageunique}
If $ G_{a, b} $ is precise, the player with advantage has a unique equilibrium strategy.
\end{theorem}

\begin{proof}
Let player $A$ have advantage. Suppose that $S_A$ and $S_A'$ are distinct equilibrium strategies for player $A$. Let $S_A$ and $S_A'$ have lengths $\ell$ and $\ell'$ respectively. By the Reverse Theorem, player $B$ has strategies $\mathcal{R}(S_A)$ and $\mathcal{R}(S_A')$ which have lengths $\ell$ and $\ell'$ respectively. Suppose $\ell' \ne \ell$. Assume, without loss of generality, that $\ell' > \ell$. Then $\mathcal{R}(S_A')$ is a Nash equilibrium strategy for $B$ with length greater than $S_A$ which contradicts Propostion \ref{grounded_gapfree}. Thus, $\ell = \ell'$. 

Assume, without loss of generality, that $(M_AS_A)_{\ell} \ge (M_AS_A')_{\ell}$. That is, we assume, that if $B$ bids $\ell$ against $S_A$ he will do no better than if he were bidding $\ell$ against $S_A'$. It is possible he will do strictly worse as bidding $\ell$ is not necessarily a part of player $B$'s optimal strategy. Consider the following function:
\[ S(x) = S_A' + x(S_A - S_A') \]
We claim that for any $x$ for which $S(x)$ is a valid strategy, $S(x)$ is an optimal strategy. Note that $S(x)$ has entrywise sum of $1$ so $S(x)$ is at least valid for $x \in [0,1]$. Consider:
\[ (M_AS(x))_i = (M_AS_A')_i + x(M_AS_A - M_AS_A')_i \]
For $i < \ell$, $(M_AS_A')_i = (M_AS_A)_i = v_A$ so $ (M_AS(x))_i = v_A$. For $i = \ell$, $(M_AS_A)_i \geq (M_AS_A')_i$ so $(M_AS(x))_i \geq (M_AS_A')_i \geq v_A$. If player $B$ bids anything greater than $\ell$ then he will do strictly worse than if he bid $\ell$, because he will win by more than he would by bidding $\ell$. Therefore, $S(x)$ guarantees player $A$ a payoff of at least $v_A$. Choose the maximal $x^\star$ for which $S(x^\star)$ is valid. Because $S(x)$ has entrywise sum of $1$, it is only invalid if $S(x)$ has a negative entry. Thus, at this maximal $S(x^\star)$ has at least one zero entry. Either $S(x^\star)$ has length less than $\ell$, a $0$ in its first entry, or is not gap-free. Each of these is impossible (above, Prop \ref{grounded_gapfree}). Therefore distinct optimal strategies $S_A$ and $S_A'$ cannot exist.
\end{proof}

In most precise games, both players have unique optimal strategies. It is possible, however, to construct a game in which the player without advantage has multiple optimal strategies. We give a characterization of these as well. If $S$ is a strategy let $(0,S)$ represent a new strategy where anytime one would bid $i$ in $S$ he will bid $i+1$ in $(0,S)$.

\begin{theorem}[]\label{noadvantagestrategies}
Let $ G_{a, b} $ be precise and let player $A$ have advantage. The following statements hold:
\begin{enumerate}[(1)]
\item Player $B$ has a unique strategy of minimal length. This strategy is $\mathcal{R}(S_A)$.
\item If Player $B$ has more than one optimal strategy, then another optimal strategy is of the form $(0, \mathcal{R}(S_A))$. 
\item All other optimal strategies for player $B$ are of the form 
\[t\mathcal{R}(S_A) + (1-t)(0, \mathcal{R}(S_A)) \hspace{0.5cm} t \in [0,1].\]
\end{enumerate}
\end{theorem}

\begin{proof}

Throughout this proof we will use a method from the proof of Theorem \ref{advantageunique}. Suppose we have two strategies $P$ and $T$ such that wherever $P$ is non-zero so is $T$. Then we define
\[ E(x) = T + (P-T)x \]
We showed above that $E(x)$ gives an optimal strategy as long as it is valid. If we choose $x^*$ to be maximal so that $E(x^*)$ is valid, then $E(x^*)$ gives an optimal strategy with a $0$ in some spot where $S$ was nonzero. Let us call the strategy produced by this method $E(P,T)$.

We begin with (1). By the Reverse Theorem, player $B$ has a strategy $\mathcal{R}(S_A)$ which is of the same length as $S_A$. By Proposition \ref{grounded_gapfree}, player $B$ cannot have a strategy shorter than $S_A$. Therefore, $\mathcal{R}(S_A)$ is a strategy of minimal length for player $B$. Suppose $S$ is another strategy of minimal length for player $B$. Then $S^{\ast} = E(S,\mathcal{R}(S_A))$ is either of lesser length, is not gap-free, or has a $0$ in the first entry. The first two possibilites are impossible by Proposition \ref{grounded_gapfree}. In the third case, we can apply the same method again to get $E(S, S^{\ast})$ which is either of lesser length, not gap free, or has $0$'s in the first two entries. Each of these is impossible by Proposition \ref{grounded_gapfree}.

We now proceed to (2). Suppose player $B$ has more than one optimal strategy. Then by (1) it must be of length greater than $\mathcal{R}(S_A)$. Let $\ell$ be the length of $\mathcal{R}(S_A)$. By Proposition \ref{grounded_gapfree}, any other optimal strategy of player $B$ must be of length $\ell + 1$. Let $S$ be such a strategy. Suppose $S_0 \ne 0$. Then we can take $S' = E(\mathcal{R}(S_A),S)$ which must have a $0$ in the first coordinate lest we contradict Proposition \ref{grounded_gapfree}. We must show that $S' = (0, \mathcal{R}(S_A))$. Because $M_B$ is Toeplitz, 
\[\left(M_B \cdot (0, \mathcal{R}(S_A))\right)_{i+1} = \left(M_B \cdot \mathcal{R}(S_A)\right)_i \]
Therefore $(0, \mathcal{R}(S_A))$ guarantees player $B$ at least his optimal payoff unless player $A$ plays $0$. Suppose that if player $A$ bids $0$ then $(0, \mathcal{R}(S_A))$ gives player $B$ a payoff of $v$ less than his optimal payoff of $v_B$. Then define a strategy,
\[ S^\triangle = \frac{S' - c(0, \mathcal{R}(S_A))}{1-c} \]
for $c$ sufficiently small so that $S' - c(0, \mathcal{R}(S_A))$ has all positive entries. Then $S^\triangle$ is a valid strategy that guarantees player $B$ his optimal payoff if player $A$ bids anything from $1$ to $\ell+1$. It guarantees player $B$ more than his optimal payoff if player $A$ bids $0$ as:
\[ \left(M_B \cdot \frac{S' - c(0, \mathcal{R}(S_A))}{1-c}\right)_0 = \frac{1}{1-c} \cdot (v_B - cv) >\frac{1}{1-c} \cdot (v_B - cv_B) = v_B \]
$S^\triangle$ is a strictly better strategy than $S'$ as player $A$ always bids $0$ with nonzero probability. $S'$ is optimal so this is impossible. Thus, $(0, \mathcal{R}(S_A))$ is an optimal strategy. That it is equal to $S'$ will follow from (3).

Finally we prove (3). $\mathcal{R}(S_A)$ and $(0, \mathcal{R}(S_A))$ are optimal strategies so any convex combination of the two is optimal. Let $S^\star$ be an optimal strategy for player $B$ that is not a convex combination of the two. Then,
$S^\star$ must be of length $\ell+1$. Therefore we can take $E((0, \mathcal{R}(S_A)), S^\star)$. This gives a strategy which is either of length $\ell$, is not gap-free, or has multiple $0$'s at the begining. The latter two possibilities are impossible by Proposition \ref{grounded_gapfree}. $\mathcal{R}(S_A)$ is the unique optimal strategy of length $\ell$ so:
\[\mathcal{R}(S_A) = (0, \mathcal{R}(S_A)) + x(S^\star -(0, \mathcal{R}(S_A)))\]
\[\frac{1}{x}\mathcal{R}(S_A) + \frac{(x-1)}{x}(0, \mathcal{R}(S_A)) = S^\star\]
Note that $\frac{1}{x} +  \frac{(x-1)}{x} = 1$ and both coefficients must be postive or else the first or last entry of $S^\star$ will be negative. Thus, $S^\star$ is a convex combination of $\mathcal{R}(S_A)$ and $(0, \mathcal{R}(S_A))$.
\end{proof}

\section{Imprecise Games}

\subsection{Adjustments for Precision}
In most of the above proofs we assume that $G_{a,b}$ is a precise game. In many games with small associated graphs, this is not the case. The simplest example is a game where in the associated graph the only directed edge goes to $\mathcal{A}$. Then player $A$ always wins, so the chip counts do no matter whatsoever.
Thus, we apply a small adjustment to the payoff matrices for players $A$ and $B$. Pick a small $x>0$. We now define $M_A^x(G_{a,b})$ as
\[ M_A^x(G_{a,b}) = M_A(G_{a,b}) + xB_{a,b} \]
where $B_{a,b}$ is given by the $(b+1) \times (a+1)$ Toeplitz matrix
\[ \left[ \begin{array}{cccc} 
a & a-1 & \cdots & 0 \\
a+1 & a & \cdots & 1 \\
\vdots & \vdots & \ddots & \vdots \\
a+b & a+b-1 & \cdots & b
\end{array} \right]. \]
Intuitively, we can think of $xB_{a, b}$ as a payoff matrix that gives payoff $x$ for each chip a player has at the end of a turn.
$S_A^x(G_{a,b})$ is then given by the strategy that maximizes player $A$'s minimum payoff under $M_A^x(G_{a,b})$. $v_A^x(G_{a,b})$ is this payoff.

While the payoff no longer corresponds exactly to winning probability, the game $G_{a, b}^x$ is still zero-sum, with total utility $1+(a+b)x$ split between the two players. We generalize our Lemma \ref{basics} to this new game:

\begin{lemma}
The game represented by $M_A^x(G_{a,b})$ is precise.
\end{lemma}

\begin{proof}
Each entry of $M_A^x(G_{a,b})$ represents a successor state of the game where each player has some number of chips. From the way we have defined $B_{a,b}$, for any successor state in which having one more chip provided an equal payoff in $G_{a,b}$, having one more chip will now provide a payoff exactly $x$ greater.
\end{proof}

A natural question arising from this adjustment is whether or not it gives a good approximation of the actual payoff for $G_{a,b}$ and the actual Nash equilibria. The following theorem shows that by choosing a small enough $x$, $M_A^x, v_A^x$, and $S_A^x$ can be made arbitarily close to $M_A, v_A$ and some Nash equilbrium strategy $S_A$.

\begin{theorem}
\label{convergence}
With $S_A$ as described above,
\begin{equation} \lim_{x \rightarrow 0} M_A^x(G_{a,b}) = M_A(G_{a,b}),\tag{1}\end{equation}
\begin{equation} \lim_{x \rightarrow 0} v_A^x(G_{a,b}) = v_A(G_{a,b}),\tag{2}\end{equation}
\begin{equation} \lim_{x \rightarrow 0} S_A^x(G_{a,b}) = S_A(G_{a,b}).\tag{3}\end{equation}
\end{theorem}

\begin{proof}[Proof of (1) and (2).]
We notice that (1) follows directly from the definition of $M_A^x$:
$$\lim_{x \rightarrow 0} M_A^x(G_{a,b}) = \lim_{x \rightarrow 0} (M_A(G_{a,b}) + xB_{a,b}) = M_A(G_{a,b})$$

We now consider (2). We can define three functions:
\begin{align*}
v_A^x(G_{a,b}) 	&= \min_i(M_A^x(G_{a,b})\cdot S_A^x(G_{a,b}))_i = g(x)\\
				\\
v_A^x(G_{a,b}) 	&= \min_i(M_A^x(G_{a,b})\cdot S_A^x(G_{a,b}))_i\\
				&= \min_i(M_A(G_{a,b})\cdot S_A^x(G_{a,b}) + xB\cdot S_A^x(G_{a,b}))_i\\
                &\le \min_i(M_A(G_{a,b})\cdot S_A^x(G_{a,b}))_i + \max_i(xB\cdot S_A^x(G_{a,b}))_i\\
                &\le \min_i(M_A(G_{a,b})\cdot S_A(G_{a,b}))_i + \max_i(xB\cdot \mathbf{1})_i\\
                &= v_A(G_{a,b}) + \max_i(xB\cdot \mathbf{1})_i = h(x)\\
                \\
v_A^x(G_{a,b}) 	&= \min_i(M_A^x(G_{a,b})\cdot S_A^x(G_{a,b}))_i\\
				&\ge \min_i(M_A^x(G_{a,b})\cdot S_A(G_{a,b}))_i = f(x)
\end{align*}
Notice that for all $x \ge 0$, $f(x) \le g(x) \le h(x)$. We also see that
$$\lim_{x \rightarrow 0} f(x) = \lim_{x \rightarrow 0} \min_i(M_A^x(G_{a,b})\cdot S_A(G_{a,b}))_i = \min_i(M_A(G_{a,b})\cdot S_A(G_{a,b}))_i = v_A(G_{a,b})$$
$$\lim_{x \rightarrow 0} h(x) = \lim_{x \rightarrow 0} v_A(G_{a,b}) + \max_i(xB\cdot \mathbf{1})_i = v_A(G_{a,b}) +\max_i(B\cdot \mathbf{1})_i \cdot \lim_{x \rightarrow 0} x  = v_A(G_{a,b})$$
Therefore, 
\[\lim_{x \rightarrow 0}h(x) = \lim_{x \rightarrow 0} v_A^x(G_{a,b}) = v_A(G_{a,b}).\qedhere\]
\end{proof}

This leaves (3), the proof of which is more nuanced. We must first develop some more theory of all-pay bidding games.

\subsection{Restricted Games}
In many bidding games, the random distribution governing optimal play does not involve bidding above some threshold. In a game of Bidding Tic-Tac-Toe where each player begins with 100 chips, a player should not bid 100 on the first turn. By the Reverse Theorem, the two players, have optimal strategies of equal length. Suppose in some bidding game $G_{a,b}$, both players have strategies of length $\ell$. Then we can consider the \textbf{restricted} game, $G_{a,b} \mid \ell$, where both players can bid at most $\ell - 1$ on the first turn and play returns to normal thereafter. In such a restricted game players are still able to play the length $\ell$ optimal strategy they would have employed in the original game. Is this strategy still optimal?

\begin{lemma}
\label{restrictedgame}
If $S_A, S_B$ are optimal length $\ell$ strategies in $G_{a,b}$ that provide the payoffs $v_A$ and $1-v_A$ respectively, then they are optimal in $G_{a,b} \mid \ell$ and provide the same payoffs.
\end{lemma}

\begin{proof}
$M_A(G_{a,b} \mid \ell)$ is the $\ell \times \ell$ top-left minor of $M_A(G_{a,b})$ as the games are identical after the first move. Thus, both players bidding less than $\ell$ in $G_{a,b}$ is equivalent to the players making the same bids in $G_{a,b} \mid \ell$. Thus, $M_A(G_{a,b} \mid \ell) \cdot S_A$ gives the first $\ell$ entries of $M_A(G_{a,b}) \cdot S_A$. The minimum entry of $M_A(G_{a,b}) \cdot S_A$ is $v_A$ so the minimum entry of $M_A(G_{a,b} \mid \ell) \cdot S_A$ is at least $v_A$. Thus, $S_A$ guarantees at least the payoff $v_A$. Using the same logic for $M_B(G_{a,b} \mid \ell)$, we obtain the $S_B$ guarantees the payoff at least $1-v_A$. The total payoff is exactly $1$ so player $A$ gets payoff $v_A$ and cannot do better and player $B$ gets the payoff $1-v_A$ and cannot do better.
\end{proof}

Furthermore, recall that precision is a characteristic of the successor states in a game. The possible successors of a restricted game are a subset of the successors of the normal game. Thus, if a game is precise then its restricted game is also precise. We are now able to state a powerful result for the restricted game that will allow us to prove some important results for general bidding games.

\begin{lemma}
\label{preciseinvertible}
In a precise game $G_{a,b}$, if $S_A$, an optimal strategy of minimal length, has length $\ell$, then $M_A(G_{a,b} \mid \ell)$ is invertible.
\end{lemma}

\begin{proof}
Suppose by way of contradiction that there exists $y \in \mathbb{R}^\ell$ such that $M_A(G_{a,b} \mid \ell) \cdot y = 0$. Define $\bar{y} \in \mathbb{R}^a$ by $\bar{y}_i = y_i$ for $0 \le i \le \ell-1$ and $\bar{y}_i = 0$ for $i \ge \ell$. Then $M_A(G_{a,b}) \cdot \bar{y}$ is a vector with $0$ in it first $\ell$ entries. In particular, $(M_A(G_{a,b}) \cdot \bar{y})_0 = 0$. $S_A$ has all positive entries so there exists $c \in \mathbb{R}$ such that $S_+ = S_A + c\bar{y}$  and $S_- = S_A - c\bar{y}$ have all positive entries. We note that:
\[ (M_A(G_{a,b}) \cdot S_+)_i  = (M_A(G_{a,b}) \cdot S_A)_i + (M_A(G_{a,b}) \cdot c\bar{y})_i = v_A + 0 = v_A \]
\[ (M_A(G_{a,b}) \cdot S_-)_i  = (M_A(G_{a,b}) \cdot S_A)_i - (M_A(G_{a,b}) \cdot c\bar{y})_i = v_A + 0 = v_A \]
for $0 \le i \le \ell-1$. 
Suppose the sum of the entries of $S_+$ is less than $1$. Then there exists $k > 1$ such that the sum of the entries of $kS_+$ is equal to $1$. Then $kS_+$ is a valid strategy for player $A$. that gives payoff $kv_A > v_A$ against player $B$'s first $\ell$ pure strategies. Thus, against $\mathcal{R}(S_A)$, $kS_+$ is better than $S_A$ so $(S_A, \mathcal{R}(S_A))$ is not a Nash equilibrium. Contradiction. Then suppose the sum of the entries of $S_+$ is greater than $1$. Then the sum of the entries of $S_-$ is less than $1$ so the same argument holds. Then suppose the the sum of the entries of $S_+$ equals $1$. Then $S_+$ and $S_A$ are optimal in $G_{a,b} \mid \ell$. $G_{a,b} \mid \ell$ is precise, however, so there exists only one optimal strategy of minimal length for either player in $G_{a,b} \mid \ell$. Therefore $y$ must equal $0$. 
\end{proof}

A method for computing optimal strategies for the player with advantage, say player $A$, now becomes apparent. Given the length of the player's unique optimal strategy we can consider the payoff matrix of the restricted game. By the Reverse Theorem, player $B$ has a gap-free strategy of the same length. Then the restricted payoff matrix multiplied by player $A$'s optimal strategy must give a constant vector. The inverse of our restricted payoff matrix multiplied by some non-zero constant vector will therefore give a scalar multiple of player $A$'s optimal strategy. 

\begin{theorem}
\label{strategyformula}
Let player $A$ have advantage. In a precise game $G_{a,b}$ if $S_A$ has length $\ell$ then
\[ S_A = \frac{M_A(G_{a,b} \mid \ell)^{-1} \mathbf{1}}{\mathbf{1}^T M_A(G_{a,b} \mid \ell)^{-1} \mathbf{1}} \]
\end{theorem}

\begin{proof}
As discussed above $M_A(G_{a,b} \mid \ell)^{-1} \cdot \mathbf{1}$ is a scalar multiple of $S_A$. The sum of the entries of $S_A$ is $1$ so we need only divide by the sum of the entries of $M_A(G_{a,b} \mid \ell)^{-1} \cdot \mathbf{1}$. This is given by $\mathbf{1}^T M_A(G_{a,b} \mid \ell)^{-1} \mathbf{1}$. 
\end{proof}

This theorem gives an explicit and rapid method for computing optimal strategies for a player with advantage. Combined with the Reverse Theorem, we will be able to develop a method for computing optimal strategies for both players in any simple bidding game. First, we will return to (3) of Theorem \ref{convergence}.

\subsection{Convergence of Strategies}
Recall our conjecture that as $x \rightarrow 0$, $S_A(G_{a,b}^x) \rightarrow S_A(G_{a,b})$. The above theorem gives even more weight to this claim as together they give a method for approximating optimal strategies for imprecise games via a convergent sequence of strategies for precise games.

We begin by partially extending the invertibility of the restricted payoff matrix to imprecise games. The importance of this result is not immediately obvious, but it will be integral to the proof of part (3) of Theorem \ref{convergence}. For simplicity, we will sometimes write $M_A(G_{a,b} \mid \ell)$ as $M_A(\ell)$ and $M_B(G_{a,b} \mid \ell)$ as $M_B(\ell)$.

\begin{prop}
\label{wow}
If player $A$ has a length $\ell$ optimal strategy for $M_A^x = M_A(G^x_{a,b})$ then at least one of $M_A(G_{a,b} \mid \ell)$ and $M_B(G_{a,b} \mid \ell) = \mathbf{1} - M_A(G_{a,b} \mid \ell)^T$ is invertible.
\end{prop}

\begin{proof}
For simplicity, let $M_A = M_A(G_{a,b} \mid \ell)$ and $M_B = M_B(G_{a,b} \mid \ell)$. If $M_A$ is invertible, we are done, so suppose $M_A$ is not invertible. Let $w \ne 0$ be in the nullspace of $M_A$. 
Because $S_A^x$ is gap-free, there exists $c > 0$ sufficiently small such that $S_A^x \pm cw$ are valid strategies for player $A$. Then
\[
M_A^x \cdot (S_A^x \pm cw) = M_A^xS_A^x \pm (M_A + xB)\cdot cw = M_A^xS_A^x \pm cxBw.
\]
Each successive row in $B$ is 1 greater in each entry than the previous row. Suppose that the sum of the entries of $w$ is equal to $0$. Then,
\[ (Bw)_{i+1}=(Bw)_i+(1, \ldots, 1)\cdot w = (Bw)_i \] Thus, $Bw$ is a constant vector. If $Bw  = 0$ then
\[ M_A^x \cdot w = M_Aw + xBw = 0.\]
By Lemma $\ref{preciseinvertible}$, $M_A^x$ is invertible so $Bw$ cannot equal $0$. Therefore, either $S_A^x + cw$ or $S_A^x - cw$ results in a better payoff for player $A$ than $S_A^x$ for $M_A^x$ contradicting the optimality of $S_A^x$. Therefore the sum of the entries of $w$ is not $0$.

We can then let the sum of the entries of $w$ be equal to $1$. Then
\[w^T(\mathbf{1}-M_A^T) = (1, \ldots, 1).\]
We will return to $w$ momentarily. We can compute that $M_B^x = \mathbf{1}-M_A^T+xB$.  Since $S_A^x$ is a Nash equilibrium, 
\[ (S_A^x)^T \cdot M_B^x = (v, \ldots, v) \]
\[ (S_A^x)^T \cdot (\mathbf{1}-M_A(\ell)^T) + x(S_A^x)^TB-(v, \ldots, v) = 0 \]
\[ (S_A^x)^T \cdot (\mathbf{1}-M_A(\ell)^T) +(d+x(\ell-1)-v, d+x(\ell-2)-v, \ldots, d-v) = 0 \]
where $d=x(S_A^x)^T\cdot(0, 1, \ldots, \ell-1)^T$. Then we can substitute $w$ into the equation:
\begin{align*}
x(\ell-1, \ldots, 1, 0) &= -(S_A^x)^T\cdot(\mathbf{1}-M_A(\ell)^T)+(v-d)(1, \ldots, 1) \\
&= -(S_A^x)^T\cdot(\mathbf{1}-M_A(\ell)^T)+(v-d)w^T(\mathbf{1}-M_A(\ell)^T)\\
&= (-(S_A^x)^T+(v-d)w^T)(\mathbf{1}-M_A(\ell)^T).
\end{align*}
Let $r_0=x(-(S_A^x)^T+(v-d)w^T)$ and $r_i=r_0+xiw^T$ so that:
 \[ r_i(\mathbf{1}-M_A(\ell)^T)=x(\ell-1+i, \ldots, 1+i, i) \]
Let $R$ be a $\ell \times \ell$ matrix with rows $r_0, \ldots, r_{\ell-1}$. Then
\[ R(\mathbf{1}-M_A(\ell)^T)=xB. \]
Therefore through this seemingly arbitrary construction we obtain that 
\[ (I_{\ell \times \ell}+R)(\mathbf{1}-M_A(\ell)^T)=\mathbf{1}-M_A(\ell)^T+xB=M_B^x\]
which is invertible by Lemma \ref{preciseinvertible}. Thus, $\mathbf{1}-M_A(\ell)^T$ is invertible.
\end{proof}

The last several results have dealt with payoff matrices of restricted games. The payoff matrix of a restricted game is, by definition, dependent on the length of a player's optimal strategy. The following lemma further demonstrates the relevance of the lengths of the players' optimal strategies.

\begin{lemma}
\label{weakconvergence}
If there exists $\ell_0$ and $x_1 > x_0 \ge 0$ such that for all $x_0 < x < x_1$, $\ell(S_A(G_{a,b}^x)) = \ell_0$ then
\[ \lim_{x \rightarrow x_0} S_A(G_{a,b}^x) \]
exists and is an optimal strategy.
\end{lemma}

\begin{proof}
Let $F_{a,b} = G_{a,b}^{x_0}$. We will treat $F_{a,b}$ as imprecise so the proof holds for both precise and imprecise games. By Proposition \ref{wow}, at least one of $M_A = M_A(F_{a,b} \mid \ell_0)$ and $M_B = M_B(F_{a,b} \mid \ell_0)$ is invertible. Suppose $M_A$ is invertible. Then the limit
\[ \lim_{x \rightarrow 0} S_A(F_{a,b}^x) = \lim_{x \rightarrow 0} \frac{(M_A^x)^{-1} \mathbf{1}}{\mathbf{1}^T (M_A^x)^{-1} \mathbf{1}} = \frac{(M_A)^{-1} \mathbf{1}}{\mathbf{1}^T (M_A)^{-1} \mathbf{1}} = S\]
exists. As $x$ goes to 
$0$, $S_A^x(F_{a,b})$ is nonzero and has entry wise sum of $1$. Thus, $S$ is all nonnegative and also has entry wise sum of $1$. Finally,
\[ v_A = \lim_{x\rightarrow 0} v_A^x = \lim_{x \rightarrow 0} \min(M_A^x \cdot S_A^x) = \min(M_A \cdot S) \]
Thus $S$ is optimal. If $M_A$ is not invertible, then $M_B$ is invertible. By the Reverse Theorem, for all $x_0 < x < x_1$, $\ell(S_B(G_{a,b}^x)) = \ell$. Therefore, we can apply the same argument as above to $S_B^x(G_{a,b})$. 
\end{proof}

While the above lemma's potential power is clear, we have not yet demonstrated that the conditions it requires are met by any games. We need some restrictions on the length of optimal strategies as we adjust chip value in order to effectively use the above results. The next lemma and its corollary give us the necessary structure.

\begin{lemma}
\label{openneighborhoods}
Let player $A$ have advantage. Let $\ell_0 = \max_{x \in \mathbb{R}_{>0}} \ell(S_A^x)$. The set of $p$ such that $\ell(S_A^p) = \ell_0$ is open in $\mathbb{R}_{>0}$. 
\end{lemma} 

\begin{proof}
The length of $S_A^x$ is an integer and is bounded above by $a$. Hence $\ell_0$ exists. Pick some $p$ so that $\ell(S_A^p) = \ell_o$. Suppose there exists no $\epsilon, \delta > 0$ such that for all $p' \in N_{\epsilon, \delta}(p) = (p-\epsilon, p+\delta)$, we have $\ell(S_A^{p'}) = \ell_0$. Then we can define a sequence $\{x_k\} \rightarrow p$ by $x_k \in N_{1/k,1/k}(p)$ so that $\ell(S_A^{x_k}) < \ell_0$. There exist only a finite number of possible values for $\ell(S_A^x)$ so there must be at least one $\ell_1 < \ell_0$ such that$\{x_k\}$ has a convergent subsequence $\{x_{a_k}\}$ with $\ell(S_A^{x_{a_k}}) = \ell_1$ for all $k$.

By Theorem \ref{convergence},
\[ \lim_{k \rightarrow \infty} M_A^{x_{a_k}}(G_{a,b} \mid \ell_1 ) = M_A^p(G_{a,b} \mid \ell_1),\]
\[ \lim_{k \rightarrow \infty} v_A^{x_{a_k}}(G_{a,b} \mid \ell_1 ) = v_A^p(G_{a,b} \mid \ell_1).\]
Then
\[
\lim_{k\rightarrow \infty} S_A^{x_{a_k}}=
\lim_{k \rightarrow \infty} ((M_A^{x_{a_k}}(\ell_1))^{-1} \cdot (v_A^{x_{a_k}}\mathbf{1}_{\ell_1}))=
(M_A^p(\ell_1))^{-1}\cdot (v_A^p\mathbf{1_{\ell_1}}) = S
\]
for which we have that
\[ \min(M_A^p \cdot S) =  \lim_{k \rightarrow \infty} \min(M_A^{x_{a_k}} \cdot S_A^{x_{a_k}}) = \lim_{k \rightarrow \infty} v_A^{x_k} = v_A^p \]
$S$ is then an optimal strategy in $G^p_{a,b}$. $S$ is the limit of length $\ell_1$ strategies so it has length at most $\ell_1$. Therefore $S \ne S_A^p$. The player with advantage has exactly one optimal strategy so an appropriate open neighborhood must exist.
\end{proof}

\begin{lemma}
Let $\ell_0$ be as above and let $M_A$ be invertible. Also assume $S_A^x$ is of constant-length on some interval $(a,b)$. Then there exists vectors $S,T$ such that for all $x\in (a,b)$,
\[ S_A^x = S + xT. \]
\end{lemma}

\begin{proof}
Let 
\[ S = \frac{M_A^{-1}\mathbf{1}}{\mathbf{1}^TM_A^{-1}\mathbf{1}}. \]
Note that $S$ is not necessarily optimal or even a valid strategy. It satisfies two notable properties. The sum of the entries of $S$ is $1$ and $M_AS$ is a constant vector. Consider,
\[ (M_A^x - M_A)(S_A^x - S) = xB(S_A^x - S)\]
$xB(S_A^x - S)$ is a constant vector as each row in $B$ differs by a vector of all $1$'s from the row above it. A vector of all $1$'s multiplied by $S_A^x - S$ is $0$ as both $S_A^x$ and $S$ have entrywise sum of $1$. Let this constant vector be denoted $\mathbf{u}$. Then
\[ (M_A^x - M_A)(S_A^x - S) = \mathbf{u} \]
\[ M_A^xS_A^x + M_AS - M_AS_A^x - M_A^xS = \mathbf{u} \]
\[ M_A^x S_A^x + M_A S - M_A S_A^x - M_AS - xBS = \mathbf{u}\]

Note the $M_A S$ terms cancel, and that $M_A^x S_A^x$ is a constant vector. Thus, because $\mathbf{u}$ is also a constant vector, we know that $M_A S_A^x + xBS$ is a constant vector, which we call $\mathbf{v}$.
Then
\[ M_AS_A^x + xBS = \mathbf{v} \]
\[ S_A^x = (M_A)^{-1}(\mathbf{v} - xBS) \]
Note that $M_A^{-1}\mathbf{v}$ is a scalar multiple of $S$. Let this scalar be $c$. We have the relation:
\[ S_A^x = cS - xM_A^{-1}BS \]
We see that $c$ is a function of $x$, and must be the unique scalar that causes $cS - xM_A^{-1}BS$ to have entrywise sum of $1$. Thus $c$ is given by:
\[ \sum_{i=0}^a (cS - xM_A^{-1}BS)_i = 1 \]
\[ c\sum_{i=0}^a S = 1 + x\sum_{i=0}^a(M_A^{-1}BS)_i \]
\[ c = 1 + x\sum_{i=0}^a(M_A^{-1}BS)_i \]
$\sum_{i=0}^a(M_A^{-1}BS)$ is a constant because $M_A, B$, and $S$ are. Let it be denoted $r$.
\[ S_A^x = (1+rx)S - xM_A^{-1}BS= S + x(rS - M_A^{-1}BS) \]
$rS - M_A^{-1}BS$ is a vector independent of $x$. Let it be denoted by $T$. Thus,
\[ S_A^x = S + xT \] 
Thus, on an $x$-interval on which $S_A^x$ is of constant length, $S_A^x$ is given by $S+xT$. Further, each entry of $S_A^x$ is given by a linear equation $S_i + xT_i$. 
\end{proof}

Note that the above lemma does not make use of anything specific to player $A$ or $B$. Thus, it also applies to $S_B^x$ if the necessary conditions hold.

\begin{corollary}
\label{lengthstable}
Let player $A$ have advantage. Let $\ell_0$ be as above. Then there exists $x_0 > 0$ such that for all $0 < x \le x_0$, $\ell(S_A^x) = \ell_0$. 
\end{corollary}

\begin{proof}
Let $x \in \mathbb{R}_{>0}$ be chosen such that $\ell(S_A^x) = \ell_0$. By Proposition \ref{wow}, at least one of $M_A$ and $M_B$ is invertible. Suppose first that $M_A$ is invertible. By Lemma \ref{openneighborhoods}, there exists an open interval $(a,b)$ containing $x$ on which $S_A^x$ is constant-length. Let $(a,b)$ be the largest such open interval. We are able apply the above lemma. There exists vectors $S,T$ such that for all $x \in (a,b)$ 
\[ S_A^x = S + xT. \]
For $x > 1$, the value of a chip is greater than the value of winning the game so neither player will ever bid more than $0$. Thus, $b \le 1$. Suppose that $a > 0$. On this interval $S_A^x$ is given by $S+xT$ for some $S,T$. Therefore, the $(\ell_0)$-th entry of $S_A^x$ is either strictly increasing, strictly decreasing, or constant. By Lemma \ref{weakconvergence} and the uniqueness of Nash equilibrium strategies for the player with advantage.
\[ \lim_{x \rightarrow a} S_A^x = S_A^a\]
\[ \lim_{x \rightarrow b} S_A^x = S_A^b\]
If $S_A^a$ or $S_A^b$ have length $\ell_0$ then by Lemma \ref{openneighborhoods} there is an open interval about $a$ or $b$ respectively on which optimal strategies have length $\ell_0$ so $(a,b)$ is not maximal. Thus, both $S_A^a$ and $S_B^b$ must have length less than $\ell_0$. This implies that:
\[ \lim_{x \rightarrow a} (S_A^x)_{\ell_0 - 1} = 0 = S_i + aT_i \]
\[ \lim_{x \rightarrow b} (S_A^x)_{\ell_0 -1 } = 0 = S_i + bT_i\]
A linear equation has at most one zero unless $S_i$ and $T_i$ are both $0$. $S_i + xT_i = (S_A^x)_{\ell_0} \ne 0$ however so this cannot be the case. Therefore $a$ must be equal to $0$. Then $S_A^x$ is constant-length on some interval which has $0$ as an endpoint.

Now suppose that $M_B$ is invertible. We can perform the same operations on the optimal strategy of minimal length $S_B^x$ for player $B$ and then apply the Reverse Theorem to achieve the same result for player $A$. 
\end{proof}

Given this structure, we can complete our discussion of convergence.

\begin{theorem}
Let player $A$ have advantage. Then
\[ \lim_{x \rightarrow 0} S_A(G_{a,b}^x) = S_A(G_{a,b}) \]
exists and is optimal.
\end{theorem}

\begin{proof}
By Corollary \ref{lengthstable}, there exists $x_0$ such that for all $0 < x \le x_0$, $\ell(S_A^x) = \ell_0$. These are the necessary conditions to apply Lemma \ref{weakconvergence} which gives the result.
\end{proof}

\section{Computing the Optimal Strategy}

Although we have developed results on the structure of optimal bidding in all-pay bidding games, we have yet to fully describe how these optimal strategies can be found.  In this section, we put together our results for precise games with our convergence results for imprecise games to give an algorithm to calculate the optimal bidding strategy for any state in an all-pay bidding game.

\subsection{Main Algorithm}
In this section, we will discuss the algorithm we developed to quickly calculate an optimal strategy.
Our algorithm first assigns to each chip an arbitrarily small but positive value $x$. This adjusted game is precise, so we will be able to take advantage of the structure we have shown for precise games. In particular, we will be able to use Theorem \ref{strategyformula}, which gives a formula for the unique bidding strategy belonging to the player with advantage, in terms of the payoff matrix and optimal length:
\[ S_A = \frac{M_A(G_{a,b} \mid \ell)^{-1} \cdot \mathbf{1}}{\mathbf{1}^T \cdot M_A(G_{a,b} \mid \ell)^{-1} \cdot \mathbf{1}} \]
From the convergence results in the previous section, the resulting strategy will be able to approximate an optimal strategy for player $A$ in an imprecise game to any desired degree of accuracy. Note that this strategy is not guaranteed to be a unique optimal strategy in the unadjusted game if the unadjusted game is not precise.
Once $S_A$ is known, we know by convergence that $\mathcal{R}(S_A)$ will have to be an optimal strategy for player $B$. $(S_A, \mathcal{R}(S_A))$ is then within any desired degree of accuracy of a Nash equilibrium for the unadjusted game.

For now we will assume the payoff matrix is known. Then, to implement Theorem \ref{strategyformula} we just need to invert the appropriate minor of that matrix, multiply by a vector of $1$'s, and rescale so that the entries of the resulting vector sum to $1$.  The problem now is to find this optimal length in a precise game where the payoff matrix is given.  The next two lemmas will allow us to use binary search to find the optimal length quickly.

\begin{lemma}
\label{lengthsearch1}
Let the game be precise. Let $\mathbf{1}_k$ be a vector of all $1$'s.  Then for all $1 \leq k \leq \ell$, $M_A(k)^{-1} \cdot \mathbf{1}_k$ will have all nonnegative entries.
\end{lemma}

\begin{proof}
By similar reasoning as in Lemma \ref{preciseinvertible}, we know that $M_A(k)^{-1}$ will be invertible for all $k \leq \ell$.
We naturally consider the game $G_{a, b} \mid k$.
Let $S_A^k$ and $S_B^k$ be $A$'s and $B$'s optimal strategies in this game.
Note that if $k=\ell$, then by definition of $\ell$, we have that $M_A(\ell)\cdot S_A$ gives a constant (nonnegative) vector,
so $M_A(\ell)^{-1} \cdot \mathbf{1}$ will be $S_A$ scaled by $1/v_A$.
This will have all nonnegative entries because $S_A$ is a strategy.
We can extend this reasoning to when $k<\ell$ if we know that $S_A^k$ still has length $k$, as it must also give some constant payoff, $v_A^k$, in $G_{a, b} \mid k$.

Suppose $S_A^k$ does not have length $k$.  Then $S_A^k$ has length $m < k < \ell$.  Let $v_A^k$ be the value of $G_{a, b} \mid k$ for player $A$.
Suppose that $v_A^k \geq v_A$.  Then, we can make strategy $S_A'$ for player $A$ in $G_{a,b}$, by extending $S_A^k$ to the full game, where $(S_A)_i=(S_A^k)_i$ if $i\leq m-1$, and is $0$ otherwise.
Then, note that $(M_A \cdot S_A')_i=v_A^k \geq v_A$ if $i \leq m-1$.
Because $m<k$, $m \leq k-1$ where $k-1$ is the maximal number of chips useable in the $G_{a, b} \mid k$ game.
Thus, if $i=m$, $(M_A \cdot S_A')_m \geq v_A^k \geq v_A$ by definition of Nash Equilibrium for $G_{a, b} \mid k$.  This is $A$'s payoff against $S_A'$ if $B$ purely bids $m$.

But since $A$'s maximal bid in $S'$ is $m-1$, that means if $B$ uses a pure strategy where she bids $i>m$ chips, she will just be winning the same bids by more chips, which cannot be better in any way.
Thus, $(M_A \cdot S_A')_i \geq (M_A \cdot S_A')_m \geq v_A^k \geq v_A$.
Thus, for all $0<i<\ell-1$, $(M_A \cdot S_A')_i \geq v_A$, so $S_A'$ is a Nash Equilibrium for $G_{a, b}$ as well.
But $S_A'$ has length $m < \ell$, so it would have to be distinct from $S_A$ because it has a different length.  This cannot be the case as $A$'s optimal strategy is unique.  Thus, we have a contradiction and $S_A^k$ cannot have length less than $k$.

Thus, $S_A^k$ has length $k$, so by the same argument as the $k=\ell$ case, all the entries of $M_A(k)^{-1} \cdot \mathbf{1}$ are nonnegative.
Note that because none of the above reasoning depended upon player $A$ having advantage, if $v_A^k < v_A$, we can apply the above argument from player $B$'s perspective. 
\qedhere
\end{proof}

\begin{lemma}
\label{lengthsearch2}
Let the game be precise. Let $\mathbf{1}_k$ be a vector of all $1$'s.  Then for all $k>\ell$, either $M_A(k)$ is not invertible or 
$M_A(k)^{-1} \cdot \mathbf{1}_k$ will have all nonnegative entries.
\end{lemma}

\begin{proof}
Assume $M_A(k)$ is invertible.

We begin by showing there is no valid length $(\ell+1)$-strategy for player $A$ that produces the same payoff for player $B$'s first $\ell+1$ pure strategies. Suppose there does exist such a strategy $S$. Let $v$ be the payoff that $S$ produces against player $B$'s first $\ell+1$ pure strategies (pure bids from $0$ up to $\ell$). Note that because $S_A$ has length $\ell$ there is a Nash equilibrium strategy $\mathcal{R}(S_A)$ of length $\ell$ for player $B$. We consider three cases:
\begin{description}
  \item[(1) $v > v_A$] \hfill \\
	Since $B$ bids at most $\ell-1$, we only need to consider the first $\ell$ coordinates of $M_A \cdot S_A$ and $M_A \cdot S$.
	By our assumption, $v>v_A$ so $S$ is strictly better than $S_A$ against $\mathcal{R}(S_A)$. Thus $S_A$ cannot be a Nash Equilibrium strategy, which is a contradiction.
  \item[(2) $v < v_A$] \hfill \\
	If $v<v_A$ then let player $B$ use the strategy $\mathcal{R}(S)$. It is easy to verify that $\mathcal{R}(S)$ produces the payoff $1-v > 1-v_a$ against player $A$'s first $\ell+1$ pure strategies. Thus, by similar reasoning as in the previous case, $\mathcal{R}(S)$ is strictly better than $\mathcal{R}(S_A)$ against $S_A$, so $\mathcal{R}(S_A)$ cannot be a Nash Equilibrium strategy, which is a contradiction.
  \item[(3) $v = v_A$] \hfill \\
	If $S=(s_0, \ldots, s_{\ell}, 0, \ldots, 0)^T$, then expanding the first $\ell+1$ coordinates of $M_A \cdot S$ results in the equations $\alpha_i s_0 + \cdots + \alpha_{i+\ell}s_{\ell} = v_A$ for $i=0, \ldots, -\ell$.
	Considering the game from player $B$'s perspective, note that $\mathcal{R}(S)$ gives $B$ a payoff of $1-v_A$ against $A$'s first $\ell+1$ strategies.
	In particular, $B$'s payoff against $A$ bidding $\ell$ will be \[(1-\alpha_{\ell})s_{\ell}+(1-\alpha_0)s_0=1-(\alpha_{\ell}s_{\ell}+\cdots+\alpha_0s_0)=1-v_A.\]
	$B$'s payoff $x$ against against $A$ bidding $\ell+1$ will be \[x=1-(\alpha_{\ell+1}s_{\ell}+\alpha_1s_0).\]
	Note that because player $A$ is winning ties, $\alpha_{\ell+1}<\alpha_{\ell}, \ldots, \alpha_1<\alpha_0$, as in each case $A$ is winning by one more chip.
	Thus, $\alpha_{\ell}s_{\ell}+\cdots+\alpha_0s_0>\alpha_{\ell+1}s_{\ell}+\alpha_1s_0$ which means
	\[1-v_A=1-(\alpha_{\ell}s_{\ell}+\cdots+\alpha_0s_0)\leq 1-(\alpha_{\ell+1}s_{\ell}+\alpha_1s_0)=x.\]
	Similarly, $B$'s payoff if $A$ bids anything greater than $k$ will be greater than $1-v_A$.
	Thus, $\mathcal{R}(S)$ is a Nash equilibrium strategy of length $k$ for player $B$.
	Note that if $k$ is at least $\ell+2$, then $\mathcal{R}(S)$ will have length at least $\ell+2$, which will be a contradiction if $S_A$ has length $\ell$.
	Since $k > \ell$, this means we must have $k = \ell+1$.
	Then, $\mathcal{R}(S)$ is a Nash Equilibrium strategy of length $\ell+1$, so by Theorem \ref{noadvantagestrategies}
	it must be of the form $\lambda (s_{\ell-1}, \ldots, s_0, 0, \ldots, 0) + (1-\lambda) (0, s_{\ell-1}, \ldots, s_0, 0, \ldots, 0)$ for $0 \leq \lambda \leq 1$.
	In turn, $S$ must be of the form $\lambda (0, s_0, \ldots, s_{\ell-1}, 0, \ldots, 0) + (1-\lambda)(s_0, \ldots, s_{\ell-1}, 0, 0, \ldots, 0)$.
	We can now write $M_A(\ell+1) \cdot S = v_A \mathbf{1}_{\ell+1}$ as $M_A(\ell+1) \cdot \lambda (0, s_0, \ldots, s_{\ell-1}) + (1-\lambda)(s_0, \ldots, s_{\ell-1}, 0) = v_A \mathbf{1}$, which can be expanded to the equation
	\[ \lambda \left( \begin{array}{c}
\alpha_1s_0+\cdots+\alpha_{\ell}s_{\ell-1} \\
\alpha_0s_0+\cdots+\alpha_{\ell-1}s_{\ell-1} \\
\vdots \\
\alpha_{-(\ell-1)}s_0+\cdots+\alpha_0s_{\ell-1} \end{array} \right)
+
(1-\lambda) \left( \begin{array}{c}
\alpha_0s_0+\cdots+\alpha_{\ell-1}s_{\ell-1} \\
\vdots \\
\alpha_{-(\ell-1)}s_0+\cdots+\alpha_0s_{\ell-1} \\
\alpha_{-\ell}s_0+\cdots+\alpha_{-1}s_{\ell-1} \end{array} \right)
=
\left( \begin{array}{c}
v_A \\
v_A \\
\vdots \\
v_A \end{array} \right)\] 

By considering the first coordinate, we get the equation
\[ \lambda(\alpha_1s_0 + \cdots + \alpha_{\ell}s_{\ell-1}) + (1-\lambda)v_A = v_A \]
so we must have $\alpha_1s_0 + \cdots + \alpha_{\ell}s_{\ell-1}=v_A$ as well. Therefore,
\[ \alpha_1s_0 + \cdots + \alpha_{\ell}s_{\ell-1}=\alpha_0s_0 + \cdots + \alpha_{\ell-1}s_{\ell-1}.\]
But since the game is precise, there must be an inequality for all the coefficients: $\alpha_1 < \alpha_0, \ldots, \alpha_{\ell} < \alpha_{\ell-1},$ so
\[ \alpha_1s_0 + \cdots + \alpha_{\ell}s_{\ell-1}<\alpha_0s_0 + \cdots + \alpha_{\ell-1}s_{\ell-1} \]
because not all the $s_i$'s are $0$.
Thus, we have a contradiction, and $k$ cannot be $\ell+1$ either.

Thus, no such strategy $S$ can exist, so if $M_A(k)$ is invertible, $M_A(k)^{-1} \cdot \mathbf{1}_k$ must have some negative terms.
\qedhere
\end{description} 
\end{proof}

We can implement the binary search algorithm as follows. Let the lower bound, $low$, start as $1$. Let the upper bound, $high$, start as $\min(a,b)+1$. 

\begin{algorithm}[H]
\DontPrintSemicolon
\textbf{Function} LSearch($M_A$, $low$, $high$)\;
\eIf{$low + 1 = high$}
	{return $low$}
    {
    	$k = (low+high)/2$\;
  		\eIf{$M_A(k)^{-1} \cdot 1_k$ is all nonnegative}{
   			return LSearch($M_A$, $k$, $high$)\;
   		}{\tcp{$M_A(k)$ is not invertible or $M_A(k)^{-1} \cdot 1_k$ has a negative entry}
   			return LSearch($M_A$, $low$, $k$)\;
  		}
   	}
 \caption{Binary Search For Length}
\end{algorithm}
By Lemmas \ref{lengthsearch1} and \ref{lengthsearch2}, this algorithm will return the length of the optimal strategy for player $A$. We can then apply our formula to directly compute player $A$'s unique optimal strategy. The reverse of this strategy is an optimal strategy for player $B$. This completes the algorithm. From our results on the convergence of strategies, this algorithm is also able to approximate, with any desired degree of accuracy, optimal strategies for imprecise games.

\subsection{Recursion on Directed Graphs}
So far, our results apply to the strategy for bidding on a single turn in an all-pay bidding game. This assumes some prior knowledge of successor game states that allows the payoff matrix to be already known.  Thus, to use our algorithm to compute Nash equilibria for any all-pay bidding game state, we need some way of first finding the payoff matrix. By noting that the payoff matrices for end states (where one player has already won) can be set as $0$ and $1$ for win and loss, we use recursion from the end states of the game to find the payoff matrix for an arbitrary turn.

Consider a combinatorial game $G$ represented as a directed graph $D=(V,E)$ with two vertices marked as $\mathcal{A}$ and $\mathcal{B}$ and a token placed at some vertex of the graph. We can think of each vertex as the starting position of a subgame of $G$. Thus for player $A$ with $a$ chips and player $B$ with $b$ chips, the token on vertex $w$, we write the game as $w_{a,b}$. Let $\mathcal{S}(w)$ give all the vertices that can be moved to from $w$.

We can compute $v_A$ as follows:
$$v_A(w_{a,b}) = \left\{\begin{array}{cc}
1 & \text{if $w = \mathcal{A}$} \\
0 & \text{if $w = \mathcal{B}$} \\
S_B^T \cdot X \cdot S_A & \text{otherwise}
\end{array}\right.$$
Then if $A$ bids $i$ and $B$ bids $j$ and $A$ makes a move then $A$'s payoff will be
\[ A(i,j) = \max_{w' \in \mathcal{S}_A(w)} v_A(w'_{a-j+i, b-i+j}) \]
because $A$ will seek to maximize his probability of winning over all of his possible sucessor states. If $A$ bids $i$ and $B$ bids $j$ and $B$ makes a move then $A$'s payoff will be
\[ B(i,j) = \min_{w' \in \mathcal{S}_B(w)} v_A(w'_{a-j+i, b-i+j})\]
because $B$ will seek to minimize $A$'s probability of winning over all of her possible sucessor states.
Therefore
$$X_{i,j} = \left\{\begin{array}{cc} 
\max\left(A(i,j), B(i,j)\right) & \text{if $i < j$ or $i=j$ and $A$ has advantage} \\
\min\left(A(i,j), B(i,j)\right) & \text{if $i > j$ or $i=j$ and $B$ has advantage} \\
\end{array} \right.$$
as each player will consider the best possible scenario if he moves and the worst possible scenario if their opponent moves. Then $S_A$, $S_B$ can computed from this payoff matrix $X$, using our algorithm from before.

Note this allows us to recurse up the directed graph from states $\mathcal{A}$ and $\mathcal{B}$, first with values for those states, then values for the states one move away (i.e. $v$ such that either $\mathcal{A}$ or $\mathcal{B} \in S(v)$), then states two moves away, and so on.

\subsection{Complexity}

An arbitary $n \times n$ matrix can be inverted in $O(n^3)$ time using the Gauss-Jordan method. There exist, however, many more efficient algorithms specific to Toeplitz matrices. In particular, the Levinson-Trench-Zohar algorithm can solve a Toeplitz system in $O(n^2)$ time \citep{Musicus1988}.

For an $n \times n$ matrix, the binary search algorithm requires $\log(n)$ iterations. Each iteration requires solving one Toeplitz system and scanning one vector for negative values. Thus, the algorithm runs in time on the order of
$\log(n) \cdot (O(n^2) + O(n)) = O(\log(n)n^2)$. Thus finding an optimal strategy and corresponding payoff for a given payoff matrix requires time on the order of $O(\log(n)n^2)$.

A simple implementation of our recursive algorithm would take time growing exponentially with the depth of $D$. We can greatly speed up this process by storing each $v_A(i,j)$ that is computed. Then when $v_A(i,j)$ must be computed again the value can be looked up rather than recomputed. In the worst case, the program must compute $v_A$ for every possible combination of chips at every vertex. Because the sum of the chips is constant, this requires at most $(a+b)\cdot|V|$ computations. Thus, the entire algorithm runs in time on the order of  $O(|V| \cdot \log(n)n^2)$ where $n=a+b$. For comparison, a linear programming algorithm to achieve the same results would require time on the order of $O(|V| \cdot n^{3.5})$ \citep{DBLP}.

\section{Acknowledgements}
We would like to thank Michael Landry and Sam Payne for exposing to us the research potential surrounding all-pay bidding games and giving guidance in the research process. We would also like to thank Aviezri Fraenkel for suggesting the idea of integrating economic and combinatorial games and Ilan Adler for his helpful comments on an earlier draft of this paper. This paper was also supported in part by NSF grant CAREER DMS-1149054.


\end{document}